\documentclass[12pt,a4paper]{article}
    \usepackage[cp1251]{inputenc}
    \usepackage[english]{babel}
\usepackage{algorithm}
\usepackage[noend]{algpseudocode}    

    \usepackage{amsmath,amssymb,amsthm}

\usepackage{tikz}
\usepackage{xypic}
\usetikzlibrary{shapes}
\usetikzlibrary{arrows.meta}
     \usepackage{graphicx}
     \usepackage{wrapfig}
  \DeclareMathOperator\CSP{CSP}

  \DeclareMathOperator\Pol{Pol}

  \DeclareMathOperator\QCSP{QCSP}

\renewcommand{\le}{\leqslant}

\theoremstyle{definition}
\theoremstyle{plain}
\newtheorem{thm}{Theorem}

\newtheorem{lem}[thm]{Lemma}

\newtheorem{cor}[thm]{Corollary}

\newtheorem*{THMReductionTHM}{Theorem~\ref{ReductionTHM}}

\newtheorem*{THMReductionToCSPTHM}{Theorem~\ref{ReductionToCSPTHM}}

\newtheorem*{THMReductionToQCSPTHM}{Theorem~\ref{ReductionToQCSPTHM}}

\newtheorem*{THMReductionToCSPNoConstantsTHM}{Theorem~\ref{ReductionToCSPNoConstantsTHM}}

\newtheorem*{THMPGPClassificationTHM}{Theorem~\ref{PGPClassificationTHM}}

\newtheorem*{LemmaCSPReductionLemma}{Lemma~\ref{CSPReductionLemma}}

\newtheorem*{LemmaPiTwoReductionLemma}{Lemma~\ref{PiTwoReductionLemma}}

\newtheorem*{LemmaReduceNumberOfQuantifiersLemma}{Lemma~\ref{ReduceNumberOfQuantifiersLemma}}

\title{The complexity of the Quantified CSP having the polynomially generated powers property.}
\author{Dmitriy Zhuk}

\begin{document}
\maketitle
\begin{abstract}
It is known that if an algebra of polymorphisms of the constraint language has the 
Polynomially Generated Powers (PGP) Property then 
the Quantified CSP can be reduced to the CSP over the same constraint language with constants.
The only limitation of this reduction is that it is applicable only for the constraint languages with constants. 
We drastically simplified the reduction and generalized it for constraint languages without constants. 
As a result, we completely classified the complexity of the QCSP for constraint languages having the PGP property. 
\end{abstract}
\section{Introduction}

Let $\Gamma$ be a finite set of relations on a set $A$.
\emph{The Constraint Satisfaction Problem over the constraint language} $\Gamma$, denoted by $\CSP(\Gamma)$, is
the following decision problem:
given a formula of the form 
$$R_{1}(v_{1,1},\dots,v_{1,n_1})\wedge
\dots
\wedge R_{s}(v_{s,1},\dots,v_{s,n_s}),
$$
where $R_{1},\dots,R_{s}\in\Gamma$ and each $v_{i,j}\in\{x_1,\dots,x_{n}\}$;
decide whether the formula is satisfiable.
A natural generalization of the CSP is the \emph{Quantified Constraint Satisfaction Problem} ($\QCSP$), 
where we allow to use both existential and 
universal quantifiers.
Formally, for a constraint language $\Gamma$,
$\QCSP(\Gamma)$
is the problem 
to evaluate 
a sentence of the form 
$$\forall x_1 \exists y_1 \dots \forall x_n \exists y_n \ R_{1}(v_{1,1},\ldots,v_{1,n_{1}})
\wedge
\dots
\wedge
R_{s}(v_{s,1},\ldots,v_{s,n_{s}}),$$
where $R_{1},\dots,R_{s}\in \Gamma$, 
and $v_{i,j}\in \{x_{1},\dots,x_{n},y_1,\dots,y_n\}$ for every $i,j$ (see \cite{BBCJK,hubie-sicomp,Meditations,QC2017}).

Unlike CSP, in general the Quantified CSP can be PSpace-hard and the complexity of $\QCSP(\Gamma)$ is not known.
In this paper we consider one of the main ideas that allows to reduce 
the QCSP to the usual CSP and therefore reduce the complexity. 
To explain the idea let us consider an easier problem. 
By $\QCSP^{\Pi_{2}}$ we denote the $\Pi_{2}$-restriction of 
the QCSP, that is, given a sentence
\begin{equation}\label{qcspinstance}
    \forall x_1 \dots \forall x_n \exists y_1 \dots \exists y_m \;\Phi
\end{equation}
where $\Phi$ is a conjunction of relations from $\Gamma$;
decide whether it holds.
To solve this problem it suffices to check that 
$\Phi$ has a solution for 
all evaluations of $x_{1},\dots,x_{n}$, hence 
it suffices to look at (the conjunction of) $|A|^n$ instances of $\CSP(\Gamma^{*})$, 
where 
$\Gamma^{*}=\Gamma\cup\{(x=a)\mid a\in A\}$. 
It is not hard to show (see \cite{AU-Chen-PGP}) that, if $A^n$ can be generated under $\Pol(\Gamma)$ from some subset $\Sigma \subseteq A^n$, 
then it is sufficient to check only evaluations from $\Sigma$, 
and to solve $|\Sigma|$ instances of $\CSP(\Gamma^{*})$. 
Hence, if $\Sigma$ can be chosen small enough and can be efficiently computed then 
this simple idea gives a polynomial reduction to the CSP. 

\textbf{Example.} Suppose $A = \{0,1\}$ and $\Gamma$ consists of linear equations in $\mathbb Z_{2}$.
Let us check that the instance (\ref{qcspinstance}) holds 
for $(x_{1},\dots,x_{n}) = (0,\dots,0)$
and $(x_{1},\dots,x_{n}) = (0,\dots,0,1,0,\dots,0)$ for any position of 1. 
To do this we just solve 
$n+1$ CSP instances of $\CSP(\Gamma^{*})$ and each of these instances 
is a system of linear equations and can be solved in polynomial time. 
If at least one of the instances does not have a solution, then the instance 
(\ref{qcspinstance}) does not hold. 
Assume that all of them are satisfiable, 
then consider the relation $\Delta(x_1,\dots,x_{n}) = \exists y_1 \dots \exists y_m \; \Phi.$
Since $\Gamma$ is preserved by $x+y+z$, $\Delta$ 
is also preserved by $x+y+z$.
Applying this operation to the tuples 
$(0,0,\dots,0),(1,0,\dots,0),(0,1,0,\dots,0), \dots, (0,0,\dots,0,1)\in \Delta$ 
coordinatewise we derive that 
$\Delta = \{0,1\}^{n}$, that is $\Delta$ contains all tuples 
and 
(\ref{qcspinstance}) holds. Thus, we build a reduction of $\QCSP^{\Pi_2}(\Gamma)$
to $\CSP(\Gamma)$ and proved that $\QCSP^{\Pi_2}(\Gamma^{*})$ is solvable in polynomial time.

This idea can be generalized as follows. 
We say that a set of operations $F$ (or an algebra $(A;F)$) has 
the \emph{polynomially generated powers (PGP)}  property
if there exists a polynomial $p(n)$ 
such that 
$A^n$ can be generated from 
$p(n)$ tuples using operations of $F$.
Another behaviour that might arise is that there is an exponential function $f$ so that the smallest generating sets 
for $A^{n}$ require size at least $f(n)$. 
In this case we say that it has the \emph{exponentially generated powers (EGP)} property. 

\begin{thm}\cite{ZhukGap2015}\label{PGPEGPDichotomy}
Every finite algebra $(A;F)$ 
either has the PGP property, or has the EGP property.
\end{thm}

Moreover, it was shown that a polynomial size generating set in the PGP case 
can be chosen very simple and efficiently computable. 
We say that 
a tuple $(a_1,\dots, a_n)$ has \emph{a switch} in the $i$-th position if
$a_{i-1}\neq a_{i}$. For example, 
the tuple $(1,1,0,2,0,0,0)$ has 3 switches.
We say that 
 an algebra $(A;F)$ is \emph{$k$-switchable} if 
 $A^{n}$ is generated from all tuples with at most $k$ switches. 
We say that 
 an algebra $(A;F)$ is \emph{switchable} if 
 it is $k$-switchable for some $k$.
  Since we have only polynomially many tuples with at most $k$ switches, 
 switchability implies the PGP property. 
 It turned our that the backward implication also holds.
 
 \begin{thm}\cite{ZhukGap2015}\label{PGPMeansSwitchable}
A finite algebra $(A;F)$ has the PGP property if and only if 
it is switchable.
\end{thm}

Thus, as a generating set of polynomial size we can take the set of all tuples with at most $k$ switches, which gives a polynomial 
reduction of $\QCSP^{\Pi_{2}}(\Gamma)$ 
to $\CSP(\Gamma^{*})$ if $\Pol(\Gamma)$ has the PGP property. 

It turns out that a similar reduction can be made for the general $\QCSP(\Gamma)$.

\begin{thm}\cite{AU-Chen-PGP,LICS2015}\label{PGPImpliesRecutionOld}
Suppose $\Gamma$ contains all constants relations
and $\Pol(\Gamma)$ has the PGP property. Then 
$\QCSP(\Gamma)$ is polynomially reducible to $\CSP(\Gamma)$.
\end{thm}

The proof of this theorem was rather complicated. 
First, in \cite{AU-Chen-PGP} Hubie Chen proved that the existence of a switchable algebra (with a slightly different type of switchability)
allows to do this reduction for any constraint language with constants.
Combining the equivalence of the two types of switchability \cite{BarnyLICS2015}
with Theorem \ref{PGPMeansSwitchable} we derive
Theorem \ref{PGPImpliesRecutionOld}. Another disadvantage of the Hubie Chen reduction is that 
it works only for constraint languages with constants.

Nevertheless, Theorem \ref{PGPImpliesRecutionOld} 
and the classification of the complexity of the CSP \cite{BulatovFVConjecture,BulatovProofCSP,MyProofCSP,ZhukFVConjecture}
give the following classification of the complexity for the PGP case with constants.

\begin{cor}\cite{AU-Chen-PGP,LICS2015}
Suppose $\Gamma$ contains all constants relations
and $\Pol(\Gamma)$ has the PGP property. 
If $\Gamma$ admits a WNU polymorphism, then 
$\QCSP(\Gamma)$ is solvable in polynomial time;
$\QCSP(\Gamma)$ is NP-complete otherwise.
\end{cor}

For many years
there was a Conjecture suggested by Hubie Chen \cite{Meditations,MFCS2017} claiming that 
$\QCSP(\Gamma)$ is PSpace-hard whenever $\Pol(\Gamma)$ has the EGP property.
If this conjecture was true, then, by Theorem \ref{PGPEGPDichotomy}, it would complete the classification of the complexity at least for constraint languages with constants.
Recently, this conjecture was disproved \cite{zhuk2020qcsp} but the PGP case remains a very important case for the study of the complexity of the Quantified CSP.

\section{Main Results}

In this paper we present a simpler proof of Theorem \ref{PGPImpliesRecutionOld}.
Moreover, our reduction works not only for constraint languages with constants but 
for any constraint languages $\Gamma$ such that $\Pol(\Gamma)$ has the PGP property.
As a result we obtain a stronger claim and a complete classification of the complexity of the 
Quantified CSP for any constraint language $\Gamma$ whose polymorphisms have the PGP property.

\begin{thm}\label{ReductionToCSPTHM}
Suppose 
$\Pol(\Gamma)$ has the PGP property.
Then 
$\QCSP(\Gamma)$ 
is polynomially reducible to 
$\CSP(\Gamma^{*})$.
\end{thm}
By $\QCSP^{\Pi_2}_{n(\forall)\le k}(\Gamma)$ 
we denote the modification of $\QCSP(\Gamma)$ where 
we allow only $\Pi_2$-sentences with at most $k$ universal quantifiers, that is, 
$\forall x_1 \dots \forall x_n \exists y_1 \dots \exists y_m \;\Phi$ for $n\le k$.

\begin{thm}\label{ReductionToQCSPTHM}
Suppose 
$\Pol(\Gamma)$ has the PGP property.
Then 
$\QCSP(\Gamma)$ 
is polynomially equivalent to 
$\QCSP^{\Pi_2}_{n(\forall)\le |A|}(\Gamma)$.
\end{thm}

Let $M$ be the matrix whose rows are all the tuples of $A^{k}$ listed in the lexicographic order.
Thus, $M$ is a matrix with $k$ columns and $|A|^k$ rows.
Let $\gamma_{1}^{k},\dots,\gamma_{k}^{k}\in A^{{|A|}^{k}}$ be the $k$ columns of $M$. 
For a relation $R$ of arity $s$ on a domain $A$ by $R^k$ we denote 
the relation on the domain $A^{k}$ defined by 
$$((a_1^1,\dots,a_1^k),(a_2^1,\dots,a_2^k),\dots,
(a_s^1,\dots,a_s^k))\in R^{k}
\Leftrightarrow
\forall i\colon (a_1^i,\dots,a_s^i)\in R.$$
Put $\Gamma^{k} = \{R^{k}\mid R\in \Gamma\}$.
Then we have the following reduction to the CSP over the domain $A^{{|A|}^{|A|}}$.

\begin{thm}\label{ReductionToCSPNoConstantsTHM}
Suppose $\Pol(\Gamma)$ has the PGP property. 
Then $\QCSP(\Gamma)$ is polynomially equivalent to 
$\CSP(\Gamma^{|A|^{|A|}}\cup\{\gamma_{1}^{|A|},\dots,\gamma_{|A|}^{|A|}\})$.
\end{thm}

This implies the following complexity classification.

\begin{thm}\label{PGPClassificationTHM}
Suppose $\Pol(\Gamma)$ has the PGP property. 
Then $\QCSP(\Gamma)$ is solvable in polynomial time if 
$\Gamma^{|A|^{|A|}}\cup\{\gamma_{1}^{|A|},\dots,\gamma_{|A|}^{|A|}\}$ admits 
a WNU polymorphism; 
$\QCSP(\Gamma)$ is NP-complete otherwise.
\end{thm}

The paper is organized as follows.
In Section 3 we formulate the main reduction 
and prove the main results of the paper using 
this reduction.
In Section 4 we prove four transformations 
of formulas 
that allow 
to remove universal quantifiers, 
reduce the number of universal quantifiers, 
and transform it into a $\Pi_{2}$-form.
In Section 5 we prove the main reduction of the paper.

\section{Main Reduction}

We will need the following transformation of a 
quantified sentence.
Suppose $0=n_{0}< n_1<n_2<\dots<n_k<n_{k+1} = n+1$ 
and 
$\Psi = \exists y_1\forall x_1\exists y_2\forall x_2 \dots
\exists y_n\forall x_{n} \Phi$.
Define a 
mapping 
$\omega_{n_1,\dots,n_{k}}\colon 
\{x_1,\dots,x_n,y_1,\dots,y_{n}\}\to 
\{x_1,\dots,x_n,y_1,\dots,y_{n},z_0,z_1,\dots,z_{k}\}$ by 
$$\omega_{n_1,\dots,n_{k}}(u) = 
\begin{cases}
y_{i},& \text{if $u = y_i$ for some $i$}\\
x_{n_{j}},& \text{if $u = x_{n_{j}}$ for some $j$}\\
z_{j},& \text{if $u = x_{i}$ and $n_{j}<i<n_{j+1}$}
\end{cases}$$

By $\omega_{n_1,\dots,n_{k}} (\Psi)$ we denote the sentence such that
\begin{enumerate}
    \item every variable $u$ is replaced by the variable 
    $\omega_{n_1,\dots,n_{k}}(u)$;
    \item the quantification $\forall x_{i}$ 
    is removed whenever 
    $\omega_{n_1,\dots,n_{k}}(x_{i})\neq x_{i}$;
    \item the outer quantification 
    $\forall z_{0}\forall z_1\dots\forall z_{k}$ 
    is added.
\end{enumerate}

The key result of this paper is as follows.

\begin{thm}\label{ReductionTHM}
Suppose 
$\Pol(\Gamma)$ is $r$-switchable, 
$\Psi= \exists y_1\forall x_1\exists y_2\forall x_2 \dots
\exists y_n\forall x_{n} \Phi$, where $\Phi$ is a conjunctive formula over $\Gamma$. 
Then 
$\Psi$ holds
if and only if 
for every $1\le n_1<n_2<\dots<n_{k}\le n$, 
$k\le r$,
the sentence
$\omega_{n_1,\dots,n_{k}}(\Psi)$ holds.
\end{thm}

Thus, we have a reduction of the Quantified CSP to polynomially many instances with at most $2r+1$ universal quantifiers.
To derive the main results we will need the following technical lemmas proved in Section \ref{TransformationsSection}.

\begin{lem}\label{CSPReductionLemma}
$\QCSP_{n(\forall)\le k}(\Gamma)$ 
is polynomially reducible to 
$\CSP(\Gamma^{*})$.
\end{lem}

\begin{lem}\label{PiTwoReductionLemma}
$\QCSP_{n(\forall)\le k}(\Gamma)$ 
is polynomially reducible to 
$\QCSP^{\Pi_{2}}_{n(\forall)\le|A|^{k}}(\Gamma)$.
\end{lem}

\begin{lem}\label{ReduceNumberOfQuantifiersLemma}
$\QCSP^{\Pi_{2}}_{n(\forall)\le k}(\Gamma)$
is polynomially reducible to 
$\QCSP^{\Pi_{2}}_{n(\forall)\le |A|}(\Gamma)$.
\end{lem}

Now we can prove the reduction to the CSP with constants.

\begin{THMReductionToCSPTHM}
Suppose 
$\Pol(\Gamma)$ has the PGP property.
Then 
$\QCSP(\Gamma)$ 
is polynomially reducible to 
$\CSP(\Gamma^{*})$.
\end{THMReductionToCSPTHM}

\begin{proof}
By Theorem~\ref{PGPMeansSwitchable}, 
$\Pol(\Gamma)$ is $r$-switchable for some $r$.
By Theorem \ref{ReductionTHM} 
an instance $I$ of 
$\QCSP(\Gamma)$ is equivalent to a conjunction of polynomially many 
instances 
$I_{1},\dots,I_{t}$ of $\QCSP(\Gamma)$ such that 
each instance has at most $2r+1$ universally quantified variables.
By Lemma \ref{CSPReductionLemma}, 
for each $I_{j}$ we can build an equivalent instance $I_{j}'$
of $\CSP(\Gamma^{*})$. 
Then $I_{1}'\wedge \dots\wedge I_{t}'$ is an instance of $\CSP(\Gamma^{*})$ 
equivalent to the original instance $I$.
\end{proof}

\begin{THMReductionToQCSPTHM}
Suppose 
$\Pol(\Gamma)$ has the PGP property.
Then 
$\QCSP(\Gamma)$ 
is polynomially equivalent to 
$\QCSP^{\Pi_2}_{n(\forall)\le |A|}(\Gamma)$.
\end{THMReductionToQCSPTHM}

\begin{proof}
By Theorem~\ref{PGPMeansSwitchable}, 
$\Pol(\Gamma)$ is $r$-switchable for some $r$.
Again, by Theorem \ref{ReductionTHM} 
an instance $I$ of 
$\QCSP(\Gamma)$ is equivalent to a conjunction of polynomially many 
instances $I_{1},\dots,I_{t}$ of $\QCSP(\Gamma)$ such that 
each instance has at most $2r+1$ universally quantified variables.
By Lemma \ref{PiTwoReductionLemma} 
for each $I_{j}$ we can build an equivalent instance $I_{j}'$
of $\QCSP^{\Pi_2}_{n(\forall)\le|A|^{2r+1}}(\Gamma)$.
We additionally require that existential variables in 
$I_{1}',\dots,I_{t}'$ to be different and universal variables 
to be the same.
Then the conjunction $I_{1}'\wedge \dots\wedge I_{t}'$ can be easily transformed into 
an equivalent instance $J$ of $\QCSP^{\Pi_2}_{n(\forall)\le|A|^{2k+1}}(\Gamma)$
by moving all universal quantifiers left.
It remains to apply Lemma \ref{ReduceNumberOfQuantifiersLemma} 
to obtain an equivalent instance of $\QCSP^{\Pi_2}_{n(\forall)\le |A|}(\Gamma)$.
\end{proof}

Recall that $\gamma_{1}^{k},\dots,\gamma_{k}^{k}\subseteq A^{{|A|}^{k}}$ are the $k$ columns  of the matrix whose rows are all the tuples of $A^{k}$ listed in the lexicographic order.
For $\alpha\in A^{{|A|}^{k}}$ by 
$\alpha^{c_1,\dots,c_{|A|}}$ we denote 
the $i$-th element of $\alpha$ where the $i$-th row of the matrix is  
$(c_1,\dots,c_{|A|})$.

\begin{THMReductionToCSPNoConstantsTHM}
Suppose $\Pol(\Gamma)$ has the PGP property. 
Then $\QCSP(\Gamma)$ is polynomially equivalent to 
$\CSP(\Gamma^{|A|^{|A|}}\cup\{\gamma_{1}^{|A|},\dots,\gamma_{|A|}^{|A|}\})$.
\end{THMReductionToCSPNoConstantsTHM}

\begin{proof}
$\Rightarrow$.  Suppose $\mathcal I$ is an instance of $\QCSP(\Gamma)$.
By Theorem~\ref{ReductionToQCSPTHM}, 
we can build an equivalent instance of 
$\QCSP^{\Pi_{2}}_{n(\forall)\le |A|}(\Gamma)$, 
that is an instance $\Psi=\forall x_{1}\dots\forall x_{|A|}\exists y_{1}\dots\exists y_{n}\;\Phi$, where $\Phi$ is a conjunction of relations from $\Gamma$. 
Let us build an equivalent instance $\mathcal J$ of $\CSP(\Gamma^{|A|^{|A|}}\cup\{\gamma_{1}^{|A|},\dots,\gamma_{|A|}^{|A|}\})$.
To obtain $\mathcal J$ from $\Phi$ we
\begin{enumerate}
    \item replace each relation $R\in \Gamma$ by $R^{{|A|}^{|A|}}\in \Gamma^{{|A|}^{|A|}}$;
    \item add constraint $\gamma_{i}(x_{i})$ for each $i\in\{1,2,\dots,|A|\}$.
\end{enumerate}
Let us show that $\Psi$ and $\mathcal J$ are equivalent.
If $\Psi$ holds, then 
$\Phi$ has a solution 
$(y_1,\dots,y_{n}) = 
(\beta_{1}^{c_1,\dots,c_{|A|}},\dots,\beta_{n}^{c_1,\dots,c_{|A|}})$
for any evaluation of 
$(x_1,\dots,x_{|A|})=(c_1,\dots,c_{|A|})$.
If we interpret each $\beta_{i}$ as a tuple from $A^{{|A|}^{|A|}}$ 
then 
$(x_1,\dots,x_{|A|},y_{1},\dots,y_{n})
= (\gamma_{1},\dots,\gamma_{|A|},
\beta_{1},\dots,\beta_{n})$ is 
a solution of 
$\mathcal J$.
Similarly, if $\mathcal J$ has a solution 
$(x_1,\dots,x_{|A|},y_{1},\dots,y_{n})
= (\gamma_{1},\dots,\gamma_{|A|},
\beta_{1},\dots,\beta_{n})$
then $(y_1,\dots,y_{n}) = 
(\beta_{1}^{c_1,\dots,c_{|A|}},\dots,\beta_{n}^{c_1,\dots,c_{|A|}})$
is a solution of $\Phi$ corresponding to $(x_1,\dots,x_{|A|}) = 
(c_1,\dots,c_{|A|})$.

$\Leftarrow$. Suppose $\mathcal J$ is an instance of 
$\CSP(\Gamma^{|A|^{|A|}}\cup\{\gamma_{1}^{|A|},\dots,\gamma_{|A|}^{|A|}\})$.
To build the quantifier-free part $\Phi$ from $\mathcal J$ we
\begin{enumerate}
    \item introduce $|A|$ new variables
    $x_1,\dots,x_{|A|}$;
    \item if $\mathcal J$ contains 
    $\gamma_{i}(z)$ and $\gamma_{j}(z)$ for some variables $z$ and $i\neq j$, then we reduce 
    $\mathcal J$ to any no-instance of $\QCSP(\Gamma)$;
    \item replace each relation $R^{{|A|}^{|A|}}\in \Gamma^{{|A|}^{|A|}}$ by $R\in \Gamma$;
    
    \item remove each constraint $\gamma_{i}(z)$ 
    and replace the variable $z$ by $x_{i}$.
\end{enumerate}
It is not hard to check that
$\mathcal J$ is equivalent to 
$\Psi=\forall x_{1}\dots\forall x_{|A|}\exists y_{1}\dots\exists y_{n}\;\Phi$.
\end{proof}

As a corollary we obtain the following result.

\begin{THMPGPClassificationTHM}
Suppose $\Pol(\Gamma)$ has the PGP property. 
Then $\QCSP(\Gamma)$ is solvable in polynomial time if 
$\Gamma^{|A|^{|A|}}\cup\{\gamma_{1}^{|A|},\dots,\gamma_{|A|}^{|A|}\}$ admits 
a WNU polymorphism; 
$\QCSP(\Gamma)$ is NP-complete otherwise.
\end{THMPGPClassificationTHM}

\section{Transformations}\label{TransformationsSection}
In this section we define four transformations of 
instances, which are probably well-known. 

\subsection{Universal quantifiers removal.}

Suppose $A = \{a_1,\dots,a_s\}$, 
$\Psi = \exists y_1 \dots\exists y_{j} \forall x\; \Phi$ is a QCSP sentence
($\Phi$ may contain quantifiers). 
Then 
$\Psi$ is equivalent to 
$$\exists y_1 \dots\exists y_{j} \exists x^{1}\dots\exists x^{s}\;
\bigwedge_{i=1}^{s} (\Phi_{i}\wedge x^{i} = a_{i}),$$ 
where $\Phi_{i}$ is obtained from $\Phi$ by replacement of 
$x$ by $x^{i}$.
If the number of universal quantifiers is bounded, we can remove all the universally quantified variables in polynomial time,  
and the result can be written in the prenex form, which is a CSP sentence.  
Thus, we have the following lemma.

\begin{LemmaCSPReductionLemma}
$\QCSP_{n(\forall)\le k}(\Gamma)$ 
is polynomially reducible to 
$\CSP(\Gamma^{*})$.
\end{LemmaCSPReductionLemma}


\subsection{Moving universal quantifiers left.}
Suppose $|A| = s$, 
$\Psi = \exists y_1 \dots\exists y_{j} \forall x\; \Phi$ is a QCSP sentence
($\Phi$ may contain quantifiers). 
Then
$\Psi$ is equivalent to 
$$\forall x^{1} \dots \forall x^{s}
 \exists y_1 \dots\exists y_{j}
\bigwedge_{i=1}^{s} (\Phi_{i}),$$ 
where $\Phi_{i}$ is obtained from $\Phi$ by replacement of 
$x$ by $x^{i}$.
In this way we can move all the universal quantifiers left starting with the right one, 
and the result can be written in the prenex form to get a QCSP sentence.  
Thus, we have the following lemma.


\begin{LemmaPiTwoReductionLemma}
$\QCSP_{n(\forall)\le k}(\Gamma)$ 
is polynomially reducible to 
$\QCSP^{\Pi_{2}}_{n(\forall)\le|A|^{k}}(\Gamma)$.
\end{LemmaPiTwoReductionLemma}

\subsection{Reducing the number of quantifiers}


Suppose $A = \{a_1,\dots,a_s\}$.
Suppose we have an instance 
$\Psi=\forall x_{1}\dots\forall x_{k}\exists y_{1}\dots\exists y_{n}\;\Phi$
of $\QCSP^{\Pi_{2}}_{n(\forall)\le k}$.
We can make it easier by reducing the number of universal 
quantifiers using the following trick.
Since $k$ is a constant, there are 
only $|A|^{k}$ possible evaluations
of $(x_1,\dots,x_{k})$ we need to check.
We do not want to use constants, that is why 
we do the following.

For each mapping $\varphi$
from $[k]:=\{1,2,\dots,k\}$ 
to $[s]:=\{1,2,\dots,s\}$ 
we define the sentence $\Psi_{\varphi}$ where each variable 
$x_{i}$ is replaced by $z_{\varphi(i)}$, 
and the universal quantification is 
$\forall z_{1}\dots\forall z_{s}$ instead of 
$\forall x_1\dots\forall x_k$.
Then 
$\Psi$ is equivalent to 
$\bigwedge\limits_{\varphi:[k]\to [s]}
\Psi_{\varphi}$, 
which can be transformed to a $\Pi_2$-sentence whose 
universally quantified variables are $\forall z_{1}\dots\forall z_{s}$.
Thus, we proved the following lemma.

\begin{LemmaReduceNumberOfQuantifiersLemma}
$\QCSP^{\Pi_{2}}_{n(\forall)\le k}(\Gamma)$
is polynomially reducible to 
$\QCSP^{\Pi_{2}}_{n(\forall)\le |A|}(\Gamma)$.
\end{LemmaReduceNumberOfQuantifiersLemma}

\subsection{Transformation to a $\Pi_2$-formula.}
Suppose 
$\Psi = \exists y_1\forall x_1\exists y_2\forall x_2 \dots
\exists y_n\forall x_{n} \Phi$.
We define a transformation 
$\zeta$ that transforms $\Psi$ 
into an equivalent $\Pi_{2}$-formula. 
Note that 
$\zeta(\Psi)$ is of exponential size on the size of $\Psi$. 

For each tuple 
$(a_1,\dots,a_n)\in A^{n}$ 
the formula
$\Phi_{a_1,\dots,a_n}$ is obtained from $\Phi$
by  
\begin{enumerate}
    \item replacement of each variable 
    $y_{i}$ by $y_{i}^{a_1,\dots,a_{i-1}}$
    \item replacement of each variable 
    $x_{i}$ by $x_{i}^{a_1,\dots,a_{i}}$
\end{enumerate}

By $\zeta(\Psi)$ we denote the sentence 
such that 
\begin{enumerate}
    \item its quantifier-free part is
$\bigwedge\limits_{(a_{1},\dots,a_{n})\in A^{n}}
\Phi_{a_1,\dots,a_n}$
\item first we universally quantify the $x$-variables, 
then we existentially quantify the $y$-variables.
\end{enumerate}

Let us show that 
$\Psi$ and $\zeta(\Psi)$ are equivalent.

\begin{lem}\label{PsiImpliesZetaPsi}
$\Psi \rightarrow \zeta(\Psi)$.
\end{lem}
\begin{proof}
Suppose 
$(f_{1}(), x_{1},
f_{2}(x_1), x_{2},
f_{3}(x_1,x_2), x_{3},...,)$ is a solution of $\Psi$.
Let us show that a solution of $\zeta(\Psi)$ can be defined by 
$$y_{i}^{a_{1},\dots,a_{i-1}} = 
f_{i}(x_{1}^{a_1},x_{2}^{a_1,a_2},x_{3}^{a_1,a_2,a_{3}},
\dots,x_{i-1}^{a_{1},\dots,a_{i-1}}).$$
In fact, each 
conjunctive formula $\Phi_{a_1,\dots,a_n}$
holds as it is just $\Phi$ with all the variables renamed.
\end{proof}

\begin{lem}\label{ZetaPsiImpliesPsi}
$\zeta(\Psi)\rightarrow \Psi $.
\end{lem}
\begin{proof}
Set 
$x_{i}^{a_{1},\dots, a_{i}} = 
a_{i}$ for every $i$ and 
$a_{1},\dots, a_{i}\in A$.
Since 
$\zeta(\Psi)$ holds, 
for each assignment of 
the $x$-variables there exists 
a proper assignment of the $y$-variables. 
We define 
a solution to $\Psi$ by 
$(f_{1}(), x_{1},
f_{2}(x_1), x_{2},
f_{3}(x_1,x_2), x_{3},...,)$, 
where 
$f_{i}(a_{1},\dots, a_{i-1}) = 
y_{i}^{a_{1},\dots, a_{i-1}}$.
To check that the conjunctive formula $\Phi$ holds
for $(x_1,\dots,x_{n})=(a_1,\dots,a_n)$ 
we consider $\Phi_{a_1,\dots,a_{n}}$.
\end{proof}

\section{Proof of Theorem \ref{ReductionTHM}}\label{ProofSection}

\begin{THMReductionTHM}
Suppose 
$\Pol(\Gamma)$ is $r$-switchable, 
$\Psi= \exists y_1\forall x_1\exists y_2\forall x_2 \dots
\exists y_n\forall x_{n} \Phi$, where $\Phi$ is a conjunctive formula over $\Gamma$. 
Then 
$\Psi$ holds
if and only if 
for every $1\le n_1<n_2<\dots<n_{k}\le n$, 
$k\le r$,
the sentence
$\omega_{n_1,\dots,n_{k}}(\Psi)$ holds.
\end{THMReductionTHM}
\begin{proof}
To obtain $\omega_{n_1,\dots,n_{k}}(\Psi)$ from $\Psi$ 
we just identify universally quantified variables and move universal quantifiers left,
hence $\omega_{n_1,\dots,n_{k}}(\Psi)$ holds whenever $\Psi$ holds. 

Assume that $\omega_{n_1,\dots,n_{k}}(\Psi)$ holds 
for every $1\le n_1<n_2<\dots<n_k\le n$. 
Let us show that $\Psi$ holds.
By Lemmas \ref{PsiImpliesZetaPsi} and \ref{ZetaPsiImpliesPsi}
$\Psi$ and $\zeta(\Psi)$ are equivalent, hence it is sufficient to 
show that 
$\zeta(\Psi)$ holds.
Let $R\subseteq A^{t}$ be the relation defined by $\zeta(\Psi)$ where all universal 
quantifiers are removed.
We assume that the order of coordinates in $R$ corresponds to 
the order of variables $x_1,x_2,\dots,x_n$.

We need to show that 
$R$ is full. Since $\Pol(R)$ is $r$-switchable, 
it is sufficient to show that all tuples with at most $r$ switches are in $R$.
Consider a tuple $\alpha$ whose switches are 
in the positions corresponding to the variables 
$x_{n_1},x_{n_2},\dots,x_{n_{k}}$, where $k\le r$.
Consider a strategy for the Existential Player in 
$\omega_{n_1,\dots,n_{k}}(\Psi)$, 
that is 
$y_{i} = f_{i}(z_{0},z_1,\dots,z_{k},x_{n_1},\dots,x_{n_{\ell(i)}})$, where 
$\ell(i)$ is the maximal number such that 
$n_{\ell(i)}<i$.
Then an assignment to the $y$-variables (in the definition of $R$)
corresponding to $\alpha$ can be defined by 
$y_{i}^{a_{1},\dots,a_{i-1}} = f_{i}(b_{0},b_1,\dots,b_{k},c_{1},\dots,c_{\ell(i)})$, 
where 
\begin{itemize}
    \item  $b_{0}$ is the value of the coordinates in $\alpha$ corresponding 
    to the $x$-variables before $x_{n_1}$;
    \item  each $b_{j}$ is the value of the coordinates in $\alpha$ corresponding 
to the $x$-variables between $x_{n_j}$ and $x_{n_{j+1}}$;
\item  $b_{k}$ is the value of the coordinates in $\alpha$ corresponding 
    to the $x$-variables after $x_{n_k}$;
    \item  each $c_{j}$ is the value of the coordinates in $\alpha$ corresponding 
to the variable $x_{n_j}^{a_1,\dots,a_{n_{j}}}$.
\end{itemize}
Since the functions $f_{1},\dots,f_{n}$ give a solution to 
$\omega_{n_1,\dots,n_{k}}(\Psi)$, we defined a proper assignment 
to $\xi(\Psi)$.
Hence $\alpha\in R$ and $R$ is full.
\end{proof}

\bibliographystyle{plain}
\bibliography{refs}
\end{document}